%% file: main.tex
\documentclass[11pt]{article}

\usepackage[letterpaper,margin=1.00in]{geometry}
\usepackage{amsmath, amssymb, amsthm, thmtools, amsfonts}
\usepackage{bbm}

\usepackage{ifthen}

\usepackage{tikz}
\usetikzlibrary{positioning,decorations.pathreplacing}

\usepackage{cite}
\usepackage{appendix}
\usepackage{graphicx}
\usepackage{color}
\usepackage{algorithm}
\usepackage[noend]{algpseudocode}
\usepackage{epstopdf}
\usepackage[textsize=tiny]{todonotes}

\usepackage{framed}
\usepackage[framemethod=tikz]{mdframed}
\usepackage[bottom]{footmisc}
\usepackage[shortlabels]{enumitem}
\setitemize{noitemsep,topsep=3pt,parsep=3pt,partopsep=3pt}
\usepackage[font=small]{caption}
\usepackage{xspace}

\usepackage{mathtools}

\newtheorem{theorem}{Theorem}[section]
\newtheorem{lemma}[theorem]{Lemma}
\newtheorem{meta-theorem}[theorem]{Meta-Theorem}

\newtheorem{definition}[theorem]{Definition}

\definecolor{darkgreen}{rgb}{0,0.5,0}
\usepackage{hyperref}
\hypersetup{
    unicode=false,          % non-Latin characters in Acrobat’s bookmarks
    colorlinks=true,        % false: boxed links; true: colored links
    linkcolor=red,          % color of internal links (change box color with linkbordercolor)
    citecolor=darkgreen,        % color of links to bibliography
    filecolor=magenta,      % color of file links
    urlcolor=cyan           % color of external links
}
\usepackage[capitalize, nameinlink]{cleveref}
\crefname{theorem}{Theorem}{Theorems}
\Crefname{lemma}{Lemma}{Lemmas}
\Crefname{observation}{Observation}{Observations}
\Crefname{equation}{}{}
\algnewcommand\algorithmicswitch{\textbf{switch}}
\algnewcommand\algorithmiccase{\textbf{case}}

\algdef{SE}[SWITCH]{Switch}{EndSwitch}[1]{\algorithmicswitch\ #1\ \algorithmicdo}{\algorithmicend\ \algorithmicswitch}%
\algdef{SE}[CASE]{Case}{EndCase}[1]{\algorithmiccase\ #1}{\algorithmicend\ \algorithmiccase}%
\algtext*{EndSwitch}%
\algtext*{EndCase}%
\newenvironment{bracketenumerate}{\begin{enumerate}}{\end{enumerate}}
\providecommand{\customparagraph}[1]{\paragraph{#1}} 
\providecommand{\customsubparagraph}[1]{\paragraph{#1}}
\providecommand{\customsubsubsection}[1]{\subsection{#1}}
\input{commands.tex}

\renewcommand{\paragraph}[1]{\vspace{0.15cm}\noindent {\bf #1}:}

\usepackage{setspace}

\usepackage{mathtools}
\DeclarePairedDelimiter{\abs}{\lvert}{\rvert}

\makeatletter
\let\oldabs\abs
\def\abs{\@ifstar{\oldabs}{\oldabs*}}
\makeatother

\newcommand{\FullOrShort}{full}

\ifthenelse{\equal{\FullOrShort}{full}}{
	
  \newcommand{\fullOnly}[1]{#1}
  \newcommand{\shortOnly}[1]{}
  
  }{

    \newcommand{\fullOnly}[1]{}
    \newcommand{\IncludePictures}[1]{}
   
  }

\begin{document}

\date{}

\title{\textbf{Improved Deterministic Connectivity\\in Massively Parallel Computation}}

\author{
    Manuela Fischer \\
    \small{ETH Zurich}\\
    \small{\texttt{manuela.fischer@inf.ethz.ch}}
    \and
    Jeff Giliberti \\
    \small{ETH Zurich}\\
    \small{\texttt{jeff.giliberti@inf.ethz.ch}}
    \and
    Christoph Grunau\thanks{\scriptsize{Supported by the European Research Council (ERC) under the European Unions Horizon 2020 research and innovation programme (grant agreement No.~853109).}} \\
    \small{ETH Zurich}\\
    \small{\texttt{cgrunau@inf.ethz.ch}}
}

\date{}

\maketitle

\setcounter{page}{0}
\thispagestyle{empty}

\begin{abstract}
\input{0-abstract}
\end{abstract}

\newpage
\input{1-intro}

\input{2-preliminaries}

\input{3-matching}

\input{4-hittingset}

\input{5-connectivity}

% \newpage
\bibliographystyle{alpha}
\bibliography{references}

\appendix

\end{document}

%% file: commands.tex
\newcommand{\eps}{\varepsilon}

\newcommand{\poly}{{\rm poly}}

\providecommand{\E}{{\rm \mathbb{E}}}
\providecommand{\Var}{{\rm \mathbb{V}ar}}

\newcommand*{\congested}{\textsf{CONGESTED-CLIQUE}}
\newcommand*{\mpc}{\textsf{MPC}}

\newcommand*{\congest}{\textsf{CONGEST}}
\newcommand*{\pram}{\textsf{PRAM}}

\providecommand{\matching}{\mathcal{M}}

\providecommand{\machines}{\mathbf{M}}
\providecommand{\memory}{\mathbf{S}}

%% file: 0-abstract.tex
A long line of research about connectivity in the Massively Parallel Computation model has culminated in the seminal works of Andoni et al.~[FOCS'18] and Behnezhad et al.~[FOCS'19]. They provide a randomized algorithm for low-space \mpc\ with conjectured to be optimal round complexity $O(\log D + \log \log_{\frac m n} n)$ and $O(m)$ space, for graphs on $n$ vertices with $m$ edges and diameter $D$. Surprisingly, a recent result of Coy and Czumaj [STOC'22] shows how to achieve the same deterministically. Unfortunately, however, their algorithm suffers from large local computation time. 

We present a deterministic connectivity algorithm that matches all the parameters of the randomized algorithm and, in addition, significantly reduces the local computation time to nearly linear. 

Our derandomization method is based on reducing the amount of randomness needed to allow for a simpler efficient search. While similar randomness reduction approaches have been used before, our result is not only strikingly simpler, but it is the first to have efficient local computation. This is why we believe it to serve as a starting point for the systematic development of computation-efficient derandomization approaches in low-memory \mpc. 

%% file: 1-intro.tex
\section{Introduction}

% MPC Motivation
Due to the ever-increasing amount of data available, memory has grown to become a major bottleneck, which makes many traditional graph algorithms inefficient or even inapplicable. To overcome this obstacle, inspired by the MapReduce paradigm \cite{MapReduce}, several computation frameworks for large-scale graph processing across multiple machines have been proposed. 
The Massively Parallel Computation (\mpc) model is a clean, theoretical abstraction of these frameworks and thus serves as a basis for the systematic study of memory-restricted distributed algorithms. Introduced by Karloff et al.~\cite{KSV10} and Feldman et al.~\cite{FMS+10} in 2010, it was later refined in a sequence of works and has become tremendously popular over the past decade. 

% MPC Preliminaries
\customparagraph{\mpc\ Model} 
In the \mpc\ model, the distributed network consists of $\machines$ machines, having local memory $\memory$ each. The input is distributed across the machines and the computation proceeds in synchronous rounds. In each round, each machine performs an arbitrary \textit{local computation} and then communicates up to $\memory$ data. All messages sent and received by each machine in each round have to fit into the machine’s local space. 
The main complexity measure of an algorithm is its \textit{round complexity}, that is, the number of rounds needed by the algorithm to solve the problem. Secondary complexity measures of an algorithm are its \textit{global memory} usage---i.e., the number of machines times the memory per machine required---as well as the \textit{total computation} performed by machines to run the algorithm, i.e., the (asymptotic) sum of the local computation performed by each machine.

% MPC Low-Memory focus
We focus on the design of fully scalable graph algorithms in the \textit{low-memory} \mpc\ model, where each machine has strongly sublinear memory. More precisely, an input graph  $G = (V,E)$, with $n$ vertices and $m$ edges, is distributed arbitrarily across machines with local memory $\memory = O(n^\delta)$ each, for some constant $0 < \delta < 1$, so that the global space is $\memory_{Global} = \Omega(n + m)$.

% Graph problems framing
\customparagraph{Graph Algorithms and Connectivity} In this model, fundamental graph and optimization problems have recently gained a lot of attention. There is a plethora of work on the problems of connectivity, matching, maximal independent set, vertex cover, coloring, and many more (see, e.g., \cite{MPC-MIS-log2logn, MPC-MIS-loglogn, coyconnectivity2021, MPC-MIS-det-general, 8948671, doi:10.1137/1.9781611975482.99, CDP20ccb}). 

One particularly important (and arguably the most central) graph problem that has received increasing attention over the past few years is the one of connectivity. This is not only a problem of independent interest, but it serves as a subroutine for many algorithms. 

\begin{definition}[Connectivity Problem]
    Let $G = (V, E)$ be an undirected graph. The goal is to compute a function $cc \colon V \rightarrow \mathbb{N}$ such that every vertex $u \in V$ knows $cc(u)$ and for any pair of vertices $u, v \in V$, $u$ and $v$ are connected in $G$ if and only if $cc(u) = cc(v)$.
\end{definition}

A sequence of works \cite{andoni2018parallel, 10.1145/3293611.3331596, DBLP:journals/corr/abs-1807-10727, 8948671, 10.1145/3350755.3400249, 10.1145/3465084.3467951, 8948671} on this problem culminated in a randomized algorithm by Behnezhad et al.~\cite{8948671} that finds all connected components of a graph with diameter $D$ in $O(\log D + \log \log_{\frac m n} n)$ rounds.

In a very recent breakthrough, Coy and Czumaj~\cite{coyconnectivity2021} obtained the same round complexity with a deterministic algorithm.
Their derandomization approach, however, comes at a cost of heavy local computation, which makes it impractical for large-scale applications. 

\customparagraph{Deterministic Algorithms and Derandomization} While the problem of connectivity is of independent interest, it is instructive to view the above results in a broader context of deterministic algorithms and derandomization. 

Notably, for almost a decade, (almost) all the research in the domain of Massively Parallel Computation has focused on the study of randomized algorithms. Only recently, a sequence of works has aimed at exploring the power of the (low-memory) \mpc\ model restricted to deterministic algorithms \cite{BKM20, CDP20ccb, CDP20cc1, MPC-MIS-det-general, coyconnectivity2021}. They demonstrate that several graph problems can be solved deterministically with (asymptotic) complexity bounds that are comparable to those of the randomized algorithms. The main ingredients of these results are derandomization methods specifically tailored to the low-memory \mpc\ model: they are designed to cope with the limited memory per machine while exploiting the power of \textit{local computation} and \textit{all-to-all communication} in this setting.

This quest for efficient derandomization techniques has become one of the main problems of the area. Unfortunately, current derandomization frameworks suffer from long local running time (e.g., large polynomial or even exponential in $n^\delta$). In fact, as noted in \cite{CDP20ccb}, allowing heavy local computation might provide an advantage in the context of distributed and parallel derandomization. However, especially in  performance-oriented scenarios, local computation may quickly become a critical parameter. It thus emerges as a natural direction to study deterministic algorithms whose total computation matches that of their randomized counterparts.

\subsection{Our Contribution}

We address this issue by presenting the first computation-efficient deterministic algorithm for the problem of graph connectivity in the strongly sublinear memory regime of \mpc.
\begin{theorem}[Deterministic Connectivity]\label{thm:main}
    There is a strongly sublinear \mpc\ algorithm that given a graph with diameter $D$, identifies its connected components in $O(\log D + \log \log_{\frac m n} n)$ rounds deterministically using $O(n + m)$ global space and $\tilde{O}(m)$ total computation. 
\end{theorem}
The total computation of our algorithm significantly improves over the $\poly(n)$-bound of Coy and Czumaj~\cite{coyconnectivity2021}, with no loss in the round complexity. In fact, our algorithm matches even the state-of-the-art randomized algorithm~\cite{8948671} in all parameters up to a polylogarithmic factor in the local running time. 

While the connectivity algorithm is of independent interest, our result provides a number of other qualitative advantages. For instance, our analysis relies only on pairwise independence as opposed to the almost $O(\log n)$-wise independence of \cite{coyconnectivity2021}. Moreover, to the best of our knowledge, our result is the first that uses the framework of limited independence for derandomization without incurring a significant loss in one of the parameters (e.g., in the total computation time), and hence may be of practical interest. Furthermore, due to their simplicity, our analyses may serve as a friendly introduction to deterministic algorithms via the framework of bounded independence and, hopefully, as a stepping stone to the more systematic development of computation-efficient derandomization. 

\subsection{Randomized Connectivity Algorithms in a Nutshell}\label{randomizedConnReview}
We present the intuition of the randomized connectivity algorithms by Andoni et al.~\cite{andoni2018parallel} and Behnezhad et al.~\cite{8948671}. For a broader overview of connectivity algorithms, see \Cref{history}.

\customparagraph{Vertex Contraction}
The main idea behind connectivity algorithms working in $\tilde{O}(\log D)$ rounds is to repeatedly perform \textit{vertex contractions}~\cite{andoni2018parallel}. Contracting (often also called \textit{relabeling}) a vertex $u$ to an adjacent vertex $v$ means deleting the edge $\{u,v\}$ and connecting $v$ to all the vertices adjacent to $u$. 
The simplest way to implement this contraction-based approach is to first appoint a random subset of the vertices as \textit{leaders} (by letting each vertex independently with probability $\frac 1 2$ become a leader), and then to contract non-leader vertices to one of their leader neighbors (if any). This approach requires $O(\log n)$ rounds with high probability. 

\customparagraph{Vertex Contraction with Levels and Budgets (Andoni et al.~\cite{andoni2018parallel})}
A crucial observation to speed up the vertex contractions---going back to the graph exponentiation approach by Lenzen and Wattenhofer~\cite{LW10}---is to let each vertex expand its neighborhood to neighbors of neighbors by adding new edges (without changing the connectivity). In fact, if every vertex reaches degree $\Omega(d)$ by expanding its neighborhood in $O(\log D)$ rounds, we can mark vertices to be a leader with probability $\approx \frac{\log n}{d}$. As a result, each non-leader vertex has a leader in its neighborhood and the number of remaining vertices is $\tilde{O}(\frac{n}{d})$.

In their algorithm, Andoni et al.~\cite{andoni2018parallel} assign a level to every vertex which has not been contracted yet. 
Vertices at level $i$ have a budget of $b_i$ for expanding their neighborhood, i.e., each vertex at level $i$ can add at most $b_i$ neighbors. The initial budget $b_0$ is set to $\min(n^{\delta/2}, \sqrt{\frac{m}{n}})$ to maintain global space $O(m)$. At iteration $i$, every vertex either increases its degree to $b_i$ or finds its connected component.
As explained above, we thus can mark leader vertices with probability $\frac{\log n}{b_i}$ and perform contractions to reduce the problem size to $\tilde{O}(n/b_i)$. Hence, the budgets of remaining vertices can be updated to $b_{i+1} = b_i^{1+c}$, for a small constant $c$, while using the same global space. Overall, after $O(\log \log_{\frac m n} n)$ iterations, there will be a unique vertex left in each connected component.

\customparagraph{Random Leader Contraction (Behnezhad et al.~\cite{8948671})}
To further improve the round complexity, Behnezhad et al.~\cite{8948671} design an algorithm that applies vertex contractions and increases the budgets of vertices in an asynchronous manner, e.g., at a given time two active vertices can have different budgets. In each round, their algorithm (informally) ensures that each vertex either learns its $2$-hop neighborhood or increases its budget. We here focus on the routine that defines the budgets' increase, as this is the only step involving randomness. 

Consider the subgraph induced by vertices with budget level $i$. The crucial observation is that if a vertex has $\Omega(b_i)$ many neighbors of the same level, then contracting all of them allows us to recuperate $\Omega(b_i^2)$ budget. If each vertex is elected as a leader with probability $\approx \frac{\log n}{b_i}$, and non-leader vertices contracted to an arbitrary neighboring leader, then leaders can increase their level without exceeding the total memory.
 
\customparagraph{Increasing Initial Budget using Matching (Behnezhad et al.~\cite{8948671})}
To allow each vertex to start with a $\poly \log n$ budget, a randomized constant-round algorithm (see \cite[Algorithm 3]{8948671}) reduces the number of vertices of $G$ by a constant factor. By running it for $O(\log \log n)$ \mpc\ rounds, the problem size decreases from $n$ to $n/\poly \log n$. Intuitively, this algorithm works by contracting a constant fraction of the vertices to their lowest-ID neighbors as follows. Each vertex proposes to be contracted to its neighbor with smallest ID. A deterministic conflict resolving phase results in a graph of size $\Omega(n)$ consisting of \textit{vertex-disjoint paths}. Contracting along the edges of a constant-approximate \textit{maximum matching} in this graph with maximum degree $2$ thus allows to contract $\Omega(n)$ vertices as desired.

\subsection{Deterministic Connectivity: Comparison with the State-of-the-Art}\label{review}
We next present the main ideas behind the recent deterministic connectivity algorithm of Coy and Czumaj~\cite{coyconnectivity2021}. 

Coy and Czumaj~\cite{coyconnectivity2021} identify and extract the only two sources of randomization from the algorithms of \cite{andoni2018parallel, 8948671}, namely matching and hitting set. On the one hand, as outlined in \Cref{randomizedConnReview}, a constant approximation of matching in graphs with maximum degree $2$ can be used for the initial budget increase. On the other hand, the random leader contraction can be formulated as a variant of set cover, which we refer to as hitting set with all sets of the same size (see \Cref{def:HS} for a precise definition).

As these are the only steps involving randomness (as outlined in \Cref{randomizedConnReview}), the (efficient) derandomization of these two constant-round key algorithmic primitives immediately leads to an (efficient) deterministic connectivity algorithm. In fact, their derandomization together with the $O(\log D + \log \log n)$ randomized algorithm due to Behnezhad et al.~\cite{8948671} results in the state-of-the-art deterministic connectivity algorithm in low-memory \mpc\ \cite{coyconnectivity2021}.

Interestingly, because of the conditional lower bound framework (conditioned on the widely believed 1-vs-2-cycles conjecture for low-space \mpc\ algorithms) due to Ghaffari et al.~\cite{GKU19} and its extension to the deterministic setting due to Czumaj et al.~\cite{CDP21lb}, the two underlying problems of matching and hitting set do not admit any \textit{component-stable}\footnote{The notion of component-stability intuitively refers to the property that the choices of any vertex over the course of the algorithm are affected only by vertices in its same connected component.} constant-round deterministic algorithm. 
Hence, the authors in \cite{coyconnectivity2021} incorporate in their work derandomization techniques that are highly non-component-stable.

While their adopted derandomization framework is well-established, its efficient  implementation for obtaining a deterministic connectivity algorithm on an \mpc\ with low local space and optimal global space requires to overcome several challenges. Although the algorithm from \cite{coyconnectivity2021} achieves optimal space guarantees, the computation is suboptimal for both derandomization steps. We refine these to obtain a more efficient deterministic connectivity algorithm, as explained next. 

\customparagraph{Maximum Matching}
In \cite{coyconnectivity2021}, the problem of approximating maximum matching in graphs of maximum degree at most two is solved by searching the space of a randomized process based on pairwise independent hash functions, which are specified by $(2\log n + O(1))$ random bits. As each of the $O(n^2)$ hash functions is evaluated $O(n)$ times, with each evaluation taking $\poly \log n$ time, the resulting total computation is $\tilde{O}(n^3)$. We reduce the \textit{seed length}, i.e., the total number of random bits needed, to $O(\log \log n)$ and, as a result, obtain $\tilde{O}(n)$ total computation.

\customparagraph{Hitting Set}
For a hitting set instance with $n$ elements and a collection of $n$ subsets of size $b$, the algorithm from \cite{coyconnectivity2021} finds a hitting set of size $O(nb^{-1/5})$ by derandomizing a simple random sampling approach based on a $O(\log_b(n))$-wise $1/\poly(n)$-approximately independent family of hash functions of size $\poly(n)$. The distributed implementation of the method of conditional expectation for this process takes global space $O(nb)$ and $\poly(n)$ total computation.

We provide a low-memory \mpc\ algorithm that solves the same hitting set instance using only pairwise independent random choices with $n \cdot \poly(b)$ global space and $n \cdot \poly(b)$ total computation. Thus, the dependency on $n$ improves polynomially when $b \ll n$. It turns out that using this hitting set algorithm as a subroutine in our connectivity algorithm allows us to obtain an algorithm with total computation $\tilde{O}(m)$. 
We also note that several other works \cite{CPS17, GK18, PY18} solve the hitting set problem deterministically in the context of graph spanners in \congest\ and \congested\ using similar derandomization techniques. However, these are not straightforward to implement in the low-memory \mpc\ model.

Finally, it is worth observing that because of the shorter seeds, the \mpc\ implementation of both matching and hitting set algorithms is significantly simplified as we can perform a simple brute force search instead of using the method of conditional expectation.

\pagebreak[3]
\subsection{Further Related Work}\label{history}
The connectivity problem in low-memory \mpc\ was studied by Andoni et al.~\cite{andoni2018parallel} who presented an $O(\log D \cdot \log \log_{\frac m n} n)$ randomized algorithm, which improves upon the classic $O(\log n)$ bound derived from earlier works in the \pram\ model. Concurrently, for graphs with large spectral gap $\lambda$, i.e., $\Omega(1/\poly \log (n))$, the bound was improved in \cite{10.1145/3293611.3331596} developing a randomized $O(\log \log n + \log (1/\lambda))$ algorithm. 
% While Lacki et al.~\cite{DBLP:journals/corr/abs-1807-10727} considered random graphs, in which each edge appears independently with probability $\Omega(\log n/n)$), and proved a randomized $O(\log \log n)$ running time. 
Then, a near-optimal parallel randomized algorithm that in $O(\log D + \log \log_{\frac m n} n)$ rounds determines all connected components was developed by Behnezhad et al.~\cite{8948671}.
Subsequently, Liu et al.~\cite{10.1145/3350755.3400249} extended the same result to the arbitrary \textsf{CRCW \pram} model, which is less computationally powerful than \mpc, achieving such result with good probability\footnote{with success probability at least $1 - 1/ \poly((m \log n)/n)$}. Moreover, by developing a method that converts randomized \pram\ algorithms to highly randomness-efficient \mpc\ algorithms, Charikar et al.~\cite{10.1145/3465084.3467951} achieved a super-polynomial saving in the randomness used in \cite{8948671}, showing that $(\log n)^{O(\log D + \log \log_{m/n} n)}$ random bits suffice (with good probability), provided that the global space is $\Omega((n + m) \cdot n^\delta)$. The current deterministic state-of-the-art algorithm for connectivity is due to Coy and Czumaj \cite{coyconnectivity2021} who obtained a deterministic $O(\log D + \log \log_{\frac m n} n)$ algorithm with asymptotically optimal space. 

Finally, let us note that the connectivity problem has been studied in other regimes as well. Lattanzi et al.~\cite{LMSV11} gave a constant-round \mpc\ connectivity algorithm in the superlinear regime, i.e., each machine has local space $\Omega(n^{1+\delta})$. 
By well-known connections between linear memory \mpc\ and the \congested\ model, \cite{JN18} yields a $O(1)$-rounds randomized connectivity \mpc\ algorithm with optimal global space. 
Then, Nowicki \cite{Now21} showed that the same problem can be solved deterministically in $O(1)$ \mpc\ rounds with the same memory guarantees.

% lower bounds, component stability, one vs two cycles
On the hardness side, one of the most outstanding problems for low-space \mpc complexity is the problem of distinguishing whether an input graph is an $n$-vertex cycle or consists of two $\frac n 2$-vertex cycles (see, e.g., \cite{RVJ18,NS22} for more information). Based on the conjectured $\Omega(\log n)$ low-memory \mpc\ round-complexity lower bound for the 1-vs-2-cycles problem, Behnezhad et al.~\cite{8948671} show an $\Omega(\log D)$ lower bound for computing connected components in general graphs with diameter $D \ge \log^{1+\Omega(1)}n$. Coy and Czumaj in \cite{coyconnectivity2021} extend the same conditional lower bound to the entire spectrum of $D$ proving that no connectivity algorithm can achieve $o(\log D)$ \mpc\ round complexity.

%% file: 2-preliminaries.tex
\section{Preliminaries}

\subsection{Primitives in Low-Space \mpc}
There are a number of well-known \mpc\ primitives that will be used as black-box tools. These have been studied in the MapReduce framework and can be implemented in the \mpc\ model with stricly sublinear space per machine and linear global space. We will use the following lemma to refer to them:
\begin{lemma}[\cite{GSZ11, G99}]\label{lem:primitives}
For any positive constant $\delta$, sorting, filtering, prefix sum, predecessor, duplicate removal, and colored summation task
\footnote{Given a sequence of $n$ pairs of numbers $\langle color_i,\, x_i \rangle, i \in [n]$, with $C = \{color_i \, | \, i\in [n]\}$, compute $S_c = \sum_{i : color_i = c} x_i$ for all $c \in C$. Note that this problem can be easily solved by a constant sequence of map, shuffle, and reduce steps with $\langle color_i,\, x_i \rangle$ as key-value pairs.} 
on a sequence of $n$ tuples can be performed deterministically in MapReduce (and therefore in the \mpc\ model) in a constant number of rounds using $\memory = n^\delta$ space per machine, $O(n)$ global space, and $\tilde{O}(n)$ total computation.
\end{lemma} 
Finally, observe that these basic primitives allow us to perform all of the basic computations on graphs deterministically that we will need in a constant number of \mpc\ rounds. This includes the tasks of computing the degree of every vertex, ensuring neighborhoods of all vertices are stored on contiguous blocks of machines, sums of values among a vertex’ neighborhood, and collecting the 2-hop neighborhoods provided that they fit in the memory of a single machine.

\subsection{Derandomization Framework}\label{sec:derandomization_framework}

In this section, we give an overview of the common derandomization techniques used in all-to-all communication models \cite{CPS17, Lub93} with a focus on deterministic algorithms in the strongly sublinear memory regime of \mpc. A systematic introduction to the framework of limited independence can be found for example in \cite{Rag88, MR95, alon2016probabilistic, LW06, CW79, WC79}.

The first step is to obtain a \textit{randomized process} that produces good results in expectation based on a small search space (i.e., short random seed) by using random variables with some limited independence. 
We will use a $k$-wise independent family of hash functions, which is defined as follows:
\begin{definition}[$k$-wise independence]
    Let $N, k, \ell \in \mathbb{N}$ with $k \le N$. A family of hash functions $\mathcal{H} = \{h : [N] \rightarrow \{0, 1\}^{\ell}\}$ is $k$-wise independent if for all $I \subseteq \{1,\ldots,n\}$ with $|I| \leq k$, the random variables $X_i := h(i)$ with $i \in I$ are independent and uniformly distributed in $\{0, 1\}^\ell$ when $h$ is chosen uniformly at random from $\mathcal{H}$. If $k = 2$ then $\mathcal{H}$ is called pairwise independent. Random variables sampled from a pairwise independent family of hash functions are called pairwise independent random variables.
\end{definition}
The following is a well-known result about the existence and construction of such hash families:
\begin{lemma}[\cite{ABI86, CG89, EGL+98}]\label{lemma-hash}
    For every $N, \ell,k \in \mathbb{N}$, there is a family of $k$-wise independent hash functions $\mathcal{H} = \{h : [N] \rightarrow \{0, 1\}^{\ell}\}$ such that choosing a uniformly random function $h$ from $\mathcal{H}$ takes at most $k(\ell + \log N) + O(1)$ random bits, and evaluating a function from $\mathcal{H}$ takes time $\poly(\ell, \log N)$ time.
\end{lemma}
If there is a randomized algorithm, over the choice of a random hash function, that gives good results in expectation, one can derandomize it by finding the right choice of (random) bits. To achieve that, if the seed length is small, one can brute force it without incurring an overhead in the global space. 

In previous works this was usually not possible due to a seed length depending on $n$ of $\Omega(\log n)$ bits, which results in hash families of size larger than the space $\memory$ of a single machine. Instead, they used the method of conditional expectation or probabilities. There, one divides the seed into several parts and fixes one part at a time in a way that does not decrease the conditional expectation (or probability). This can be done with global coordination. We refer the interested reader for more details of the method of conditional expectation to \cite[Section 2.5, Appendix A]{coyconnectivity2021}.

\subsection{Reducing The Seed Length via Coloring}\label{sec:coloring} 
The following technique plays a central role for reducing the seed length of randomized processes solving \textit{local} graph problems. As showed in \cite{BKM20, CDP20ccb, MPC-MIS-det-general}, if the outcome of a vertex depends only on the random choices of its neighbors, then $k$-wise independence among random variables of \textit{adjacent} vertices is sufficient. Whenever this is the case, we can find a mapping from vertex IDs to shorter names (colors) such that adjacent vertices are assigned different names. Linial gave a $1$-round distributed coloring algorithm with $O(\Delta^2 \log(n))$ colors \cite{doi:10.1137/0221015}. 
We here adapt a more explicit $1$-round distributed coloring algorithm with $O(\Delta^2 \log_\Delta^2(n))$ colors by Kuhn \cite{10.1145/1583991.1584032} to the \mpc\ model, which leads to the following lemma:
\begin{lemma}\label{lm:coloring}
    Let $G = (V,E)$ be a graph of maximum degree $\Delta \leq n^{\delta}$. There exists a deterministic algorithm which computes an $O(\Delta^2 \log_{\Delta}^2 n)$ coloring of $G$ in $O(1)$ \mpc\ rounds using $O(n^\delta)$ local space, $O(n \cdot \poly(\Delta))$ global space, and $\tilde{O}(n \cdot \poly(\Delta))$ total computation. 
\end{lemma}
\begin{proof}
    We start by recalling the high-level idea and then we give an efficient \mpc\ implementation.
    We assume that each vertex in $G$ is given a unique ID between $1$ and $n$.
    Let $p$ be a prime with $10\Delta \log_\Delta (n) \leq p \leq 20 \Delta \log_\Delta (n)$. It is well known that such a prime always exist.
    Moreover, let $d = \lceil \log_\Delta (n) \rceil$.
    There exists $p^{d+1} \geq n$ distinct polynomials of degree at most $d$ over $\mathbb{F}_p$. We denote by $f_i$ the $i$-th such polynomial.
    Each color corresponds to a tuple over $\mathbb{F}_p$. Note that there are $p^2 = O(\Delta^2 \log_{\Delta}^2 n)$ such tuples.
    
    Let $C_i = \{(x,f_i(x)) \colon x \in \mathbb{F}_p\}$.
    Using $\Delta d < p$ together with the fact that a non-zero polynomial of degree $d$ can have at most $d$ zeros implies that each vertex can choose a color $c(i) \in C_i$ such that $c(i) \notin C_j$ for every neighbor $j$. Now, assigning each vertex $i$ the color $c(i)$ results in a valid coloring.
    It remains to discuss the MPC implementation.
    By using the basic primitives of \cref{lem:primitives} and the assumption that $\Delta \leq n^\delta$, we can assume that the machine responsible to compute the coloring of the $i$-th vertex also stores the IDs of all the neighbors of $i$. Note that a given polynomial can be evaluated in time $\poly(\log n,\Delta)$.
    Computing the color $c(i)$ boils down to $O(\Delta \cdot p^2) = \poly(\log n,\Delta)$ polynomial evaluations. Hence, the total computation time is $\tilde{O}(n \cdot \poly(\Delta))$, as desired.
\end{proof}

%% file: 3-matching.tex
\section{Constant Approximation of Maximum Matching}\label{sec:matching}
The first algorithmic step for the derandomization of the connectivity algorithm from~\cite{8948671} consists of solving approximate maximum matching in graphs of maximum degree two. Coy and Czumaj proved the following theorem: 
\begin{theorem}[Theorem 4.2 of \cite{coyconnectivity2021}]\label{coy-amm}
Let $G = (V, E)$ be an undirected simple graph with maximum degree $\Delta \leq 2$. One can deterministically find a matching $\matching$ of $G$ of size at least $m/8 = \Omega(m)$ in $O(1)$ \mpc\ rounds with local space $\memory = O(n^\delta)$, and global space $\memory_{Global} = O(n)$.
\end{theorem}
By extending their algorithm with the seed reduction technique mentioned earlier, we prove the following result:
\begin{theorem}\label{thm:matching}
    There exists an algorithm with the same properties as those in Theorem \ref{coy-amm} using $\tilde{O}(n)$ total computation.
\end{theorem}
We start by reviewing the main idea used in the algorithm proving Theorem \ref{coy-amm}.

    \customsubparagraph{Randomized Algorithm}
    The algorithm of Theorem \ref{coy-amm} is based on derandomizing the following simple random process.
    Let $\{X_e \colon e \in E\}$ be a family of pairwise independent random variables with $X_e = 1$ with probability $p = 1/4$ and $X_e  = 0$ otherwise. 
    Now, let $\matching$ be the matching that includes each edge $e$ with $X_e = 1$ and $X_{e'} = 0$ for every neighboring edge $e'$.
    The expected size of this matching is:
    \begin{align*}
        \E[|\matching|] &= \sum_{e \in E} \Pr[e \in \matching] \ge \sum_{e \in E} \Pr[X_e = 1] - \sum_{\substack{e' \in E \setminus \{e\}\::\:\\ e' \cap e \neq \emptyset}} \Pr[X_e = 1 \cap X_{e'} = 1]\\
        &\ge m \cdot (p - 2p^2) \ge \frac m 8,
    \end{align*} 
    where the second inequality follows from pairwise independence of the random variable. Hence, they can be specified by a seed of length $2\log n + O(1)$ by Lemma~\ref{lemma-hash}. As explained in \cite{coyconnectivity2021}, this allows to use the method of conditional expectation to deterministically find a matching of size at least $m/8$ in $O(1)$ \mpc\ rounds.
    
    \customsubparagraph{Reducing the Seed Length}
    We next show how one can further reduce the seed length to $O(\log \log n)$. The main observation is that the above analysis holds as long as for any two neighboring edges the two corresponding variables are independent. This motivates the following approach.
    First, we assign to each edge $e$ a color $c(e)$ from the set $\{1,2,\ldots,C\}$ for $C = O(\log^2 n)$ by applying Lemma~\ref{lm:coloring} such that two neighboring edges get assigned a different color.
    Let $\{X_c \colon c \in [C]\}$ be a family of pairwise independent random variables with $X_c = 1$ with probability $p = 1/4$ and $X_c = 0$ otherwise. 
    We now include each edge $e$ in $\matching$ if $X_{c(e)} = 1$ and $X_{c(e')} = 0$ for every neighboring edge $e'$. The same calculations as above shows that $\mathbb{E}[\matching] \geq \frac{m}{8}$. 
    
    \customsubparagraph{\mpc\ Algorithm}
    Now we are ready to present our deterministic \mpc\ algorithm that proves Theorem~\ref{thm:matching}. In the following, we say that something can be efficiently computed if there exists a deterministic \mpc\ algorithm running in $O(1)$ rounds with local space $\memory = O(n^\delta)$, global space $\memory_{Global} = O(n)$ and using $\tilde{O}(n)$ total computation. 
    
    Let $\mathcal{H} = \{h \colon [C] \mapsto \{0,1\}^2\}$ be a family of $2$-wise independent hash functions of size at most $2^{2\cdot\log C + O(1)} = \poly(\log n)$ obtained using \cref{lemma-hash}. Observe that each hash function $h \in \mathcal{H}$ defines a matching $\matching(h)$ that includes each edge $e$ with $h(c(e)) = 0$ and $h(c(e')) \neq 0$ for every neighboring edge $e'$, where $h(i)$ denotes the length-$2$ bit sequence assigned to $i$ by the corresponding integer in $\{0,\ldots,3\}$.

    The analysis of the randomized algorithm above implies that choosing a hash function $h$ uniformly at random from  $\mathcal{H}$ results in a matching  of expected size at least $m/8$. In particular, this guarantees the existence of a hash function $h^*$ with $\matching(h^*) \geq m/8$. We efficiently compute $|\matching(h)|$ for every $h \in \mathcal{H}$ and choose one \textit{good} hash function that yields a matching of size at least $m/8$.
    
    First, we efficiently compute the coloring $c$ using \cref{lm:coloring}. Next, we compute the approximate maximum matching in $G$ by derandomizing the sampling approach analyzed above. Since the size of our family of pairwise independent hash functions is $\poly \log n$, we can store one number per hash function on every machine.  Each machine $M_j$, which is responsible for some edges $\mathcal{E}_j \subseteq [E]$, can compute locally the number of edges $\matching^j(h) \subseteq E_j$ in the matching generated by $h \in \mathcal{H}$ within a single round. Then, we efficiently aggregate these numbers across all machines to compute the size of the matching $\matching(h) = \sum_j \matching^j(h)$ for every hash function $h$. The best $h^* \in \mathcal{H}$ for which $\matching(h^*) \ge \frac m 8$, breaking ties arbitrarily, yields our approximate maximum matching. Finally, let us note that the global memory occupied by the hash functions across all machines $\machines$ is $\machines \cdot |\mathcal{H}| \ll \machines \cdot O(n^\delta) = O(n)$ and the overall computation performed to evaluate each hash function for every edge is $|\mathcal{H}| \cdot \poly(\log n) \cdot O(n) = \tilde{O}(n)$.

%% file: 4-hittingset.tex
\section{Computation-Efficient Derandomization of Hitting Set}\label{sc:hs}

In this section, we give a deterministic \mpc\ algorithm for the following hitting set variant defined in \cite{coyconnectivity2021}:
\begin{definition}[\textit{Hitting Set for Leader Election}]\label{def:HS}
     Let $S_1,\ldots,S_n$ be subsets of $[n]$ with $i \in S_i$ and $|S_i| = b$, for each $i \in [n]$. The goal is to find a \textit{(small) hitting set} $\mathcal{L} \subseteq [n]$, that is, a set for which $S_i \cap \mathcal{L} \neq \emptyset$ holds for all $i \in [n]$.
\end{definition}

Coy and Czumaj~\cite{coyconnectivity2021} gave an algorithm with the same parameters as those of the random sampling approach in \cite{8948671}, except that they need large $\poly(n)$ computation. 

\begin{theorem}[Theorem 5.6 of \cite{coyconnectivity2021}]
\label{thm:hitting_set_coy}
    Let $b$ and $n$ be integers with $\log^{10} (n) \leq b \leq n$. One can deterministically find a subset $\mathcal{L} \subseteq [n]$ that solves the Hitting Set for Leader Election problem with $|\mathcal{L}| \leq O(n (\min \{ b,\memory \} )^{-1/5})$ within a constant number of \mpc\ rounds using local space $\memory = O(n^{\delta})$, global space $\memory_{Global}  = O(nb)$, and total computation $\poly(n)$.
\end{theorem}

We extend the randomized approach their algorithm relies on by using the method of alterations and reducing the amount of randomness needed to prove the following result:

\begin{theorem}\label{thm:hs}
    There exists an algorithm with the same properties as those in Theorem \ref{thm:hitting_set_coy} with two differences.
    The total computation reduces to $O(n\cdot  \poly(b))$ and the global space increases to $O(n \cdot  \poly(b))$.
\end{theorem}

We will show in \cref{sec:blackbox} that the algorithm from Theorem~\ref{thm:hs} together with minor changes to the parameters of the connectivity algorithm results in a deterministic connectivity \mpc\ algorithm with near-linear total computation.

\customsubsubsection{Review of Hitting Set Algorithm of Coy and Czumaj} Consider adding each element to $\mathcal{L}$ with probability $p = b^{-1/5}$. Assuming full independence, the assumption $b \geq \log^{10}(n)$ together with a simple Chernoff Bound implies that $\mathcal{L}$ is a hitting set with high probability. The high probability bound still holds with $O(\log_b n)$-wise independence, but  fails to hold with $o(\log_b n)$-wise independence. As $n$ $k$-wise independent random variables require a seed length of $\Omega(k \log n)$, using $O(\log_b n)$-wise independence would not result in a seed length of $O(\log n)$, which is necessary for an $O(1)$ \mpc\ round derandomization based on the method of conditional expectation. To shorten the seed length, the authors of \cite{coyconnectivity2021} use so-called $k$-wise $\eps$-approximately independent random variables for $k = 15 \log_b(n)$ and $\eps = n^{-6}$. In particular, the starting point of their algorithm is the following theorem.

\begin{theorem}[{\cite[Theorem 5.2]{coyconnectivity2021}}]\label{thm:coy-start}
Let $\log^{10}(n) \leq b \leq n$, $k$ be even with $k = 15 \log_b(n) \geq 4, \eps = n^{-6}$, and $p = b^{-1/5}$. Then, if $X_1,X_2, \ldots, X_n$ are $k$-wise $\eps$-approximately independent random variables with $X_i = 1$ with probability $b^{-1/5}$ and $X_i = 0$ otherwise. Then each of the following $n+1$ events hold with probability at least $1 - 9n^{-3}$: 

\begin{bracketenumerate}
    \item $\sum_{j \in S_i} X_j > 0$ for every $1 \leq i \leq n$, and
    \item $\sum_{i=1}^n X_i \leq 2nb^{-\frac{1}{5}}$. 
\end{bracketenumerate}
\end{theorem}

Next, we explain our randomized approach, which bears some similarities with that of \ref{thm:coy-start}, and proceed to the reduction of its seed length and its deterministic implementation on an \mpc\ with strongly sublinear memory.

\customsubsubsection{Pairwise Analysis}

As a first step, we show that a minor modification to their randomized hitting set algorithm results in a hitting set of expected size at most $2nb^{-1/5}$, assuming only pairwise independence. As before, each element joins $\mathcal{L}$ with probability $p = b^{-1/5}$. In expectation, $b \cdot p = b^{4/5}$ elements are sampled from each set. Using only pairwise independence and Chebyshev's inequality, this implies that a set is bad, i.e., no element is sampled from it, with probability at most $\frac{1}{b^{4/5}}$. This directly follows from the following lemma:

 \begin{lemma}
    Let $X_1,\ldots,X_n$ be pairwise independent random variables taking values in $[0,1]$. Let $X = X_1 + \ldots + X_n$ and $\mu = \E[X]$. Then $\Var[X] = \sum_{i=1}^n \Var[X_i] \le \mu$ and
    \begin{equation*}
        \Pr\left[|X - \mu| \ge \mu\right] \le \frac{\sum_{i=1}^n \Var[X_i]}{\mu^2} \le \frac{1}{\mu}.
    \end{equation*}
\end{lemma}

Hence, by adding for each unhit set an arbitrary element to $\mathcal{L}$, at most $n/b^{4/5}$ additional elements are added to $\mathcal{L}$ in expectation, resulting in a hitting set of expected size at most $n(b^{-1/5} + b^{-4/5})$.

\customsubsubsection{Reducing The Seed Length}

From the pairwise analysis above, we directly get a seed length of $O(\log n)$. Next, we show how to reduce the seed length to $O(\log b)$, which allows for a simple brute-force search. 
We again employ a coloring idea, which is based on the simple observation that we only require pairwise independence between elements contained in the same set. 
Hence, the goal is to color the elements with $\poly(b)$ colors such that all elements in a given set $S_i$ are colored with a different color.

In general, this may not be possible as there might exist elements which are contained in a lot of sets. Fortunately, a simple calculation shows that there exist at most $n/b$ elements which are contained in more than $b^2$ different sets. Hence, by directly adding these elements to $\mathcal{L}$, we can assume ``for free'' that each element is contained in at most $b^2$ sets, which we will do from now on.

We can then obtain a coloring with the desired properties by finding a proper coloring in the graph $G_{conflict}$, defined as follows. 
The vertex set consists of one vertex for each of the $n$ elements. Moreover, two elements are connected by an edge if there exists a set which contains both elements.
Note that the maximum degree $\Delta_{conflict}$ of $G_{conflict}$ is upper bounded by $b^3$. This follows from our assumption that each element is contained in at most $b^2$ sets.
Therefore, we can efficiently color $G_{conflict}$ with $C = O(\Delta^2_{conflict} \log^2(n)) = O(b^6 \log^2 n)$ colors.
For each $i \in [n]$, let $c(i)$ denote the color assigned to the $i$-th element.
Note that it directly follows from the definition of $G_{conflict}$ that all elements in a given set are assigned a different color.

We are now ready to present our randomized process that produces a hitting set with the desired properties.
Let $\{X_c \colon c \in [C]\}$ be a family of pairwise independent random variables with $X_c = 1$ with probability $p = b^{-1/5}$ and $X_c = 0$ otherwise. 
For simplicity, we assume that $1/p$ is a power of $2$, i.e., there exists $\ell \in \mathbb{N}$ with $2^{\ell} = b^{1/5}$. 
According to \cref{lemma-hash}, we can generate these random variables with a seed of length $2(\ell + \log C) + O(1) = O(\log b)$. Now, we add each element $i$ with $X_{c(i)} = 1$ to $\mathcal{L}$. Then, for each set $S_i$ with $\sum_{j \in S_i} X_{c(j)} = 0$, we add the element $i \in S_i$ to $\mathcal{L}$.
By the analysis and discussion above, $\mathcal{L}$ is a hitting set of expected size $O(nb^{-1/5})$.

\customsubsubsection{\mpc\ Algorithm}
It remains to discuss the \mpc\ implementation, which will prove \cref{thm:hs}. 
In the following, we say that something can be efficiently computed if there exists a deterministic \mpc\ algorithm running in $O(1)$ rounds with local space $\space = O(n^\delta)$, global space $O(n \poly(b))$, and using $O(n \poly(b))$ total computation.

In the preprocessing step, we add all elements which are contained in at least $b^2$ sets to the hitting set and remove all sets which contain at least one such element from consideration.
The preprocessing step requires us to compute for each element in how many sets it is contained in.
This can be done efficiently by using the colored summation primitive.

Next, we explain how to efficiently construct the graph $G_{conflict}$.
We generate the edges of $G_{conflict}$ in two steps. First, each set $S = \{e_1,e_2,\ldots,e_b\}$ creates $\binom{b}{2}$ entries $\{\{e_i,e_j\} \colon i \neq j \in [b]\}$.
This can easily be done with $\poly(b)$ global space per set and $\min(\memory,\poly(b))$ local space in $O(1)$ rounds by using the primitives of Lemma \ref{lem:primitives}.
Hence, we can efficiently generate all these edges in parallel.
Afterwards, we use the duplicate removal procedure of Lemma \ref{lem:primitives} to remove duplicate edges.

As $G_{conflict}$ has maximum degree $b^3$, we can use \cref{lm:coloring} to efficiently compute a coloring of $G_{conflict}$ with $C = O(b^6 \log^2 n) = \poly(b)$ colors.
As before, we denote with $c(i)$ the color assigned to the $i$-th element.
For $\ell := \log_2(b^{1/5})$, let $\mathcal{H} = \{h \colon [C] \mapsto \{0,1\}^\ell\}$ be a family of $2$-wise independent hash functions of size at most $2^{2(\ell + \log C) + O(1)} = \poly(b)$ such that evaluating a function from $\mathcal{H}$ takes time $\poly(\ell, \log C) = \poly(\log b)$ time. \cref{lemma-hash} guarantees the existence of such a family.

For each function $h \in \mathcal{H}$, we define a hitting set $\mathcal{L}_h$ as follows.
First, each element $i$ with $h(c(i)) = 0$ is contained in $\mathcal{L}_h$, where $h(c(i))$ denotes the length-$\ell$ bit sequence for $c(i)$ by the corresponding integer in $\{0,\ldots,\ell-1\}$. Moreover, if for a given set $S_i$ no element contained in it was added in the first step, then we add element $i$ to $\mathcal{L}_h$.
The discussion above implies that there exists at least one hash function $h \in \mathcal{H}$ with $|\mathcal{L}_h| = O(n b^{-1/5})$. 
Using \cref{lem:primitives}, it is easy to see that for a single hash function $h \in \mathcal{H}$, we can efficiently compute $\mathcal{L}_h$ and its size. As $\mathcal{H}$ only contains $\poly(b)$ hash functions, this implies that we can efficiently compute $\mathcal{L}_h$ for every $h \in \mathcal{H}$. After we have done this, we can output the hitting set $\mathcal{L}_{h^*}$ of smallest size. As remarked above, $\mathcal{L}_{h^*}$ has size $O(n b^{-1/5})$, which finishes the proof.

%% file: 5-connectivity.tex
\section{Connectivity Algorithm}\label{sec:blackbox}

In this section, we discuss the necessary changes to the randomized connectivity algorithm of Behnezhad et al.~\cite{8948671} and its analysis in order to prove the main result of this paper.

The deterministic approximate matching from Section~\ref{sec:matching} is used to replace steps $5$ and $6$ of Algorithm 2 of \cite{8948671}. The same modification was already done by \cite{coyconnectivity2021} and they showed that the total number of vertices drop by a constant factor, assuming that no isolated vertex exists. Hence, by applying this modified algorithm $O(\log \log_{\frac m n} n)$ times, one can in $O(\log \log_{\frac m n} n)$ rounds ensure that $m \ge n \log^{C} n$, for a given constant $C$. All the steps of the modified deterministic algorithm can be implemented by invoking the primitives of \cref{lem:primitives} $O(1)$ times, which in particular ensures that the algorithm can be implemented with total computation $\tilde{O}(m)$. Hence, we can from now on assume that $m \ge n \log^{C} n$, for a given constant $C$. 
It remains to prove that Algorithm 1 of \cite{8948671} can be implemented deterministically with the same asymptotic complexity and using $\tilde{O}(m)$ total computation, assuming $m \geq n \log^C (n)$ for a sufficiently large constant $C$. To this end, Coy and Czumaj proved the following lemma: 

\begin{lemma}[{\cite[Lemma 6.3]{coyconnectivity2021}}]\label{lem:coy-blackbox1}
    Let $S_i$ denote the set of saturated vertices at level $i$ after Step 2 of the \textsc{RelabelIntraLevel} routine in \cite{8948671}, let $L_i$ denote the set of selected leaders at level $i$ after Step 3 of the same execution of \textsc{RelabelIntraLevel}, let $\beta_i$ denote the budget of vertices at level $i$, let $b(v)$ denote the budget of vertex $v$, and let $\gamma, \eps$ be arbitrary constants such that $0 < \gamma, \eps < 1$. If we make the following modifications to \textsc{RelabelIntraLevel}:
\begin{itemize}
    \item set $\beta_{i+1} \coloneqq \beta_i \cdot (\min\{\beta_i, n^\eps\})^{\gamma/4}$,
    \item replace Step 3 of \textsc{RelabelIntraLevel} with any \mpc\ algorithm that in $O(1)$ rounds selects $O\left(\frac{|S_i|}{(\min\{\beta_i, n^\eps\})^{\gamma}}\right)$ leaders for each level $i$ with high probability or deterministically, and
    \item replace the budget update rule in Step 4 of \textsc{RelabelIntraLevel} with \[b(v) \coloneqq b(v)\cdot (\min\{b(v), n^\eps\})^{\gamma/4},\]
\end{itemize}
then the connectivity algorithm of \cite{8948671} remains correct with the same asymptotic local and global space complexity.
\end{lemma}

We extend the above lemma to make it work with the deterministic hitting set from Section \ref{sc:hs} by proving the following slight modification of it. The main technical challenge will be to ensure that our deterministic hitting set algorithm, which adds a polynomial factor (in $b$) increase in the memory and computation required, can still be run in parallel with linear global space and total computation.

\begin{lemma}\label{lem:blackbox1}
    Let $c \ge 3$ be the smallest integer such that both the global space and the total computation required by the algorithm from \cref{thm:hs} are bounded by $n \cdot b^c$, and let $\eps = \delta / c$ so that $n^{c\cdot \eps} \le n^{\delta}$. The same result as that of Lemma~\ref{lem:coy-blackbox1} can be achieved with the following modifications to \textsc{RelabelIntraLevel}:
\begin{itemize}
    \item set $\beta_{i+1} \coloneqq \beta_i \cdot (\min\{\beta_i, n^\eps\})^{\frac{\gamma}{4c}}$,
    \item replace Step 3 of \textsc{RelabelIntraLevel} with any \mpc\ algorithm that in $O(1)$ rounds selects $O\left(\frac{|S_i|}{(\min\{\beta_i, n^\eps\})^{\gamma}}\right)$ leaders for each level $i$ with high probability or deterministically using at most $n \beta_i^c$ global space and total computation, and
    \item replace the budget update rule in Step 4 of \textsc{RelabelIntraLevel} with \[b(v) \coloneqq b(v)\cdot (\min\{b(v), n^\eps\})^{\frac{\gamma}{4c}},\]
\end{itemize}
and by replacing the initial budget $\left(\frac{m}{n}\right)^{1/2}$ assigned to each vertex with $\left(\frac{m}{n}\right)^{1/2c}$ in Algorithm 1 of \cite{8948671}. Then, the connectivity algorithm of \cite{8948671} remains correct with the same asymptotic local and global space complexity.  Moreover, the resulting total computation is $O(m)$.
\end{lemma}

\begin{proof}
    We need to show that all claims and lemmas involving the modified steps of Algorithm 1 of \cite{8948671} do not affect its correctness nor its bounds on local and global memory. As in \cite{coyconnectivity2021}, we need to prove the following three key properties:
    \begin{bracketenumerate}
        \item for any vertex $v$, the value of $\ell(v)$ never exceeds $O(\log \log_{m/n} n)$ (cf. \cite[Lemma 15]{8948671}),
        \item the global space used is $O(\memory_{Global})$ (cf. \cite[Lemma 17]{8948671}),
        \item the sum of the squares of the budgets does not exceed $O(\memory_{Global})$ (cf. \cite[Lemma 21]{8948671}).
    \end{bracketenumerate}
    \begin{bracketenumerate}
        \item Recall that the budget of each vertex is increased as $\beta_{i+1} \coloneqq \beta_i \cdot (\min\{\beta_i, n^\eps\})^{\frac{\gamma}{4c}}$ and that $\beta_0 = \left(\frac{m}{n}\right)^{1/2c}$. Since the budget of any vertex cannot exceed $n$, we have that there are at most $O(\log \log_{m/n} n)$ levels as required.
        
        \item Let $n_i$ denote the number of vertices which \textit{ever} reach level $i$ over the course of the algorithm. In the proof of Lemma 17 \cite{8948671}, it is shown that the total sum of the budget increases over the course of the algorithm is $O(m)$, namely
        \begin{equation*}
            \sum_{i=1}^L \beta_i n_i = O(m).
        \end{equation*}
        We extend this claim and prove that the total sum of the global space used by all hitting set instances over all iterations of the algorithm is bounded by $O(m)$, that is
        \begin{equation*}
            \sum_{i=1}^L \beta_i^c \cdot n_i = O(m).
        \end{equation*}
        Analogously to \cite{coyconnectivity2021}, we first show that $\beta_{i+1}^c \cdot n_{i+1} \le \beta_{i}^c \cdot n_{i}$. We have that the number of vertices at level $i$ removed from the graph (i.e., not marked as a leader) per vertex marked as leader is at least:
        \begin{equation*}
            \frac{|S_i \setminus L_i|}{|L_i|} = \frac{|S_i| - |L_i|}{|L_i|} = \Omega\left((\min\{\beta_i,n^\eps\})^\gamma\right) \gg (\min\{\beta_i,n^\eps\})^{\gamma/2}.
        \end{equation*}
        It then follows that
        \begin{align*}
            \beta_{i+1}^c \cdot n_{i+1} &= \left(\beta_i \cdot (\min\{\beta_i, n^\eps\})^{\frac{\gamma}{4c}}\right)^c n_{i+1} \\ &< \left(\beta_i^c \cdot (\min\{\beta_i, n^\eps\})^{\frac{\gamma}{4}}\right)\left(n_i (\min\{\beta_i,n^\eps\})^{-\gamma/2}\right) \\ &\le \beta_{i}^c \cdot n_{i}.
        \end{align*}
        Using the fact that the maximum possible level for a vertex is $L = O(\log \log n)$, we obtain
        \begin{equation*}
            \sum_{i=1}^{L} \beta_i^c \cdot n_i \le L\cdot (\beta_0^c \cdot n_0) \le  O(\log \log n) \cdot \left(\frac m n\right)^{\frac{1}{2}} \cdot n,
        \end{equation*}
        where the last inequality comes from the fact that $\beta_0 = \left(\frac{m}{n}\right)^{\frac{1}{2c}}$. Note that we can assume that $m \geq n \log^{20c} (n)$ and therefore each vertex has an initial budget of $\beta_0  = (m/n)^{1/2c} \ge \log^{10}(n) \gg O(\log \log n)$, as required by \cref{thm:hs}. This yields
        \begin{align*}
            O(\log \log n) \cdot \left(\frac m n\right)^{\frac{1}{2}} \cdot n \ll \left(\frac m n\right)^{\frac{1}{2}} \cdot \left(\frac m n\right)^{\frac{1}{2}} \cdot n = O(m).
        \end{align*}
        
        \item Follows by the same line of reasoning as in property (2).
    \end{bracketenumerate}
    
    By the choice of $c$, repeating the same calculations as in property (b) proves that the total computation required by running our deterministic hitting set algorithm over all instances in each iteration of the algorithm does not exceed $O(m)$. Moreover, Lemma \ref{lem:primitives} implies that all the other steps of the algorithm can be implemented with total computation $\tilde{O}(m)$.
\end{proof}

We are now ready to prove our main result.
\begin{proof}[Proof of \cref{thm:main}]
    We apply \cref{lem:blackbox1} using our Hitting Set for Leader Election algorithm from \cref{thm:hs} setting $\gamma = \frac 1 5$ (Note that $m \geq n \log^C(n)$ for a sufficiently large constant $C$ implies $\beta_0 \geq \log^{10}(n)$). Then, it follows directly from Lemma 6.4 of \cite{coyconnectivity2021} combined with \cref{lem:blackbox1} that copies of our hitting set algorithms can be run in parallel, for each possible level and in a constant number of rounds within optimal global space and $\tilde{O}(m)$ total computation. Thus, we proved that all relevant aspects of the proof of correctness have been adjusted in comparison to \cite{coyconnectivity2021, 8948671}. Finally, as noted in \cite{coyconnectivity2021}, our extension of Lemma 15 in \cite{8948671} proves that the number of iterations remains asymptotically the same and that the deterministic algorithms replacing the $O(1)$-round random sampling approach take asymptotically the same number of rounds. Thus, we conclude that the round complexity is not affected. 
\end{proof}